\documentclass[11pt,times]{article}
\usepackage[utf8]{inputenc}
\usepackage{amssymb,amsmath}
\usepackage{algorithm}
\usepackage{algpseudocode}
\usepackage{amsthm}
\usepackage{enumerate}
\usepackage[font={scriptsize,it}]{caption}
\usepackage{subcaption}
\allowdisplaybreaks
\usepackage[T1]{fontenc}

\usepackage{float}
\newfloat{algorithm}{t}{lop}

\usepackage{tikz}
\usetikzlibrary{decorations.pathmorphing}
\usetikzlibrary{shapes,fit}
\usetikzlibrary{backgrounds}

\usepackage{bm}

\newtheorem{theorem}{Theorem}[section]

\newtheorem{definition}[theorem]{Definition}
\newtheorem{lemma}[theorem]{Lemma}

\newcommand{\Xomit}[1]{}

\newcommand{\rco}{resulting in a contradiction.}

\newcommand{\sxv}{score_x (v)}
\newcommand{\sv}{score (v)}

\newcommand{\la}{\leftarrow}

\newcommand{\bl}{\begin{lemma}}
\newcommand{\el}{\end{lemma}}
\newcommand{\congest}{{\sc Congest}}
\newcommand{\Tau}{H}
\newcommand{\hide}[1]{}

\newcommand{\centers}{blocker set }
\newcommand{\blocker}{blocker set }

\renewcommand\labelenumi{(\roman{enumi})}
\renewcommand\theenumi\labelenumi

\algtext*{EndWhile}
\algtext*{EndFor}
\algtext*{EndIf}

\advance \topmargin by -\headheight
\advance \topmargin by -\headsep
\evensidemargin \oddsidemargin
\begin{document}

\title{A Deterministic Distributed Algorithm for Exact Weighted All-Pairs Shortest Paths in $\tilde{O}(n^{3/2})$ Rounds}
\author{Udit Agarwal$^{\star}$, Vijaya Ramachandran\thanks{Dept. of Computer Science, University of Texas, Austin TX 78712. Email: {\tt udit@cs.utexas.edu, vlr@cs.utexas.edu}. This work was supported in part by NSF Grant CCF-1320675. }, Valerie King\thanks{Dept. of Computer Science, University of Victoria, BC, Canada. Email: {\tt val@uvic.ca}},  Matteo Pontecorvi\thanks{Nokia Bell Labs, Paris-Saclay, Nozay 91620, France. Email: {\tt matteo.pontecorvi@nokia.com}}}

\maketitle
\begin{abstract}
We present a deterministic distributed algorithm to compute all-pairs shortest paths(APSP) in an edge-weighted  directed or undirected graph. Our algorithm runs in $\tilde{O}(n^{3/2})$ rounds in the Congest model, where $n$ is the number of nodes in the graph.
This is the first $o(n^2)$ rounds deterministic distributed algorithm for the weighted APSP problem. 
Our algorithm is fairly simple and incorporates 
a deterministic distributed algorithm we develop for computing a `blocker set'~\cite{King99}, which has been used earlier in sequential
dynamic computation of APSP.
\end{abstract}

\newpage

 \section{Introduction}\label{sec:intro}

The design of distributed algorithms 
 for various network (or graph) problems such as shortest paths~\cite{LP13,Nanongkai14,Elkin17,HNS17}
 and 
 minimum spanning tree~\cite{GHS83,PR99,GKP98,KP98} is a well-studied area of research.
The most widely considered model for 
studying distributed algorithms is the \congest{}  model~\cite{Peleg00}
(also see~\cite{Elkin17,HNS17,HW12,LP13,Nanongkai14,Ghaffari15}),
described in more detail below.
In this paper we consider the problem of computing all pairs shortest paths (APSP) in a weighted
directed (or undirected) graph in this model.

The problem of computing all pairs shortest paths (APSP)
in distributed networks is a very fundamental problem, and 
there has been a considerable line of work for the \congest{} model
as described below in Section~\ref{sec:prior}. However, for a weighted
graph no deterministic algorithm was known in this model other than 
a trivial method that runs in $n^2$ rounds.
In this paper we  present the first algorithm for this problem
in the \congest{} model that computes
weighted APSP deterministically in less than $n^2$ rounds. Our algorithm computes
APSP deterministically in $O(n^{3/2} \cdot \sqrt{\log n})$ rounds in this model
in both directed and undirected graphs.

Our distributed APSP algorithm is quite simple and we give an overview in Section~\ref{sec:outline}. 
It uses the
notion of a blocker set introduced by King~\cite{King99} in the context of sequential fully dynamic
APSP computation. Our deterministic distributed algorithm for computing a blocker set is the most
nontrivial component of our algorithm, and is described in Section~\ref{sec:blocker}.
In very recent work~\cite{AR18},
these results have been incorporated in a deterministic APSP algorithm 
that runs in $\tilde{O}(n^{4/3})$ rounds.
The key to this improvement is a novel pipelined
method that improves Step~\ref{algMain:h-hop-sssp-tree} in our Algorithm~\ref{algMain}.
 
 \subsection{The {\sc Congest} Model}	\label{sec:congest}

In the {\sc Congest} model~\cite{Peleg00},
there are $n$ independent processors 
interconnected in
a network.
We refer to these processors as nodes.
These nodes are connected in the network by bounded-bandwidth links 
which we refer to as edges.
The network is modeled by a  graph $G = (V,E)$
where $V$ is the set of processors and 
$E$ is the set of edges or links between these processors.
Here $|V| = n$ and $|E| = m$.

Each node is assigned a unique ID from 
$1$ to $n$ and 
has infinite computational power.
Each node has limited topological knowledge 
and only knows about its incident edges.
For the weighted APSP problem we consider, 
 each edge has a positive integer weight
and the edge weights are  bounded by $poly(n)$.
Also if the edges are directed, 
the corresponding communication channels are bidirectional
and hence the communication network can be represented 
by the underlying undirected graph $U_G$ of $G$
(this is also considered in~\cite{HNS17,PR18,GL17}).

The computation proceeds in rounds. In each round each processor can send a message of 
size $O(\log n)$ along edges incident to it, and it receives the messages sent to it in the previous
round. The model allows a node to send different message along different edges though we do not
need this feature in our algorithm.
The performance of an algorithm in the \congest{} model is measured by its
round complexity, which is the worst-case
number of rounds of distributed communication.
Hence the goal is to minimize the round complexity 
of an algorithm.

\subsection{Prior Work} \label{sec:prior}

{\bf Unweighted APSP.}
For APSP in unweighted undirected graphs,  $O(n)$-round algorithms were given 
independently in~\cite{HW12,PRT12}. An improved $n +O(D)$-round algorithm was then given in~\cite{LP13}, where $D$ is the diameter of the undirected graph. 
Although this latter result was claimed only for
undirected graphs, the algorithm in~\cite{LP13} is also a correct $O(n)$-round APSP algorithm for
directed unweighted graphs. 
The message complexity of directed unweighted APSP was reduced to $mn +O(m)$ in a recent 
algorithm~\cite{PR18} that runs in $\min\{2n, n+O(D)\}$rounds 
 (where $D$ is now the directed diameter of the graph).
 A lower bound of
$\Omega (n/\log n)$ for the number of rounds needed to compute the diameter of the graph in the \congest{}  model
is given in~\cite{FHW12}.

{\bf Weighted APSP.}
While unweighted APSP is well-understood in the \congest{} model much remains to be done in the
weighted case. For deterministic algorithms,
weighted SSSP for a single source can be computed in $n$ rounds using the classic Bellman-Ford algorithm~\cite{Bellman58,Ford56},
and this leads to a simple deterministic weighted APSP algorithm that runs in $O(n^2)$ rounds. Nothing better was known for the number of
rounds for deterministic weighted APSP until our current results.

{\bf Exact Randomized APSP Algorithms.}
Even with randomization, nothing better than $n^2$ rounds was known for exact weighted APSP until recently,
when Elkin~\cite{Elkin17}  gave a randomized weighted APSP algorithm that runs in $\tilde{O}(n^{5/3})$ rounds
and this was further improved to $\tilde{O}(n^{5/4})$ rounds in
Huang et al.~\cite{HNS17}.
Both of these are w.h.p. results.

{\bf Deterministic Approximation Algorithms for APSP.}
There are deterministic algorithms for 
approximating
weighted all pairs shortest path problem,
and these run in $\tilde{O}(n)$ rounds for both directed~\cite{LP15} and undirected graphs~\cite{LP15,HKN16,EN16}.
 
\section{Overview of the APSP Algorithm}\label{sec:outline}

 Let $G=(V,E)$ be an edge-weighted graph (directed or undirected) with weight function $w$ and with $|V|=n$ and $|E|=m$. The \congest{} model assumes that every message is of $O(\log n)$-bit size, which restricts $w(e)$ to be an $O(\log n)$ size integer value. However,
 outside of this restriction imposed by the \congest{} model, our algorithm works for arbitrary edge-weights (even negative edge-weights as long as there is no negative-weight cycle). 
 Given a path $p$
 we will use {\it weight} or {\it distance} to denote the sum of the weights of the edges on the path and
 {\it length} (or sometimes {\it hops})  to denote the number of edges on the path.
 We denote the shortest path
 distance from a vertex $x$ to a vertex $y$ in $G$ by $\delta(x,y)$. In the following we will assume that
 $G$ is directed,
 but the same algorithm works for undirected graphs as well.
 
 An $h$-hop SSSP tree for $G$ rooted at a vertex $r$ is a tree of height $h$ where the weight of the path 
 from $r$ to a
 vertex $v$ in the
 tree is the shortest path distance in $G$ among all paths that have at most $h$ edges.
 In the case of multiple  paths with the same distance from $r$ to 
 $v$ we assume that $v$ chooses the parent vertex with minimum id as its parent in the 
 $h$-hop SSSP tree.
 We will use $h = \sqrt {n \cdot \log n}$  in our algorithm.

  Our overall APSP algorithm is given in Algorithm~\ref{algMain}. In Step~\ref{algMain:h-hop-sssp-tree} an $h$-hop 
 SSSP tree and the associated SSSP distances, $\delta_h (r,v)$, are computed at vertex $v$ for each root 
 $r\in V$.

   Step 2 computes a {\it \blocker } $Q$ of $q=O((n \log n)/h)$ 
 nodes  
  for the collection of $h$-hop SSSP
  trees constructed in Step 1.
  This step is described in detail in the  next section, where we describe a distributed implementation of
  King's sequential method~\cite{King99}. Our method computes the \blocker $Q$
   in $O(nh + (n^2 \log n)/h)$ rounds. We now 
  give the definition of a
  \centers for a collection of rooted $h$-hop trees.

\begin{algorithm}[H]	
\caption{ Main Algorithm}
\begin{algorithmic}[1]	
\State {\bf for each} $v\in V$ in sequence: compute an $h$-hop SSSP rooted at $v$ by  computing at
each $x\in V$ the $h$-hop shortest path distances $\delta_h(v,x)$, the parent pointer
 in the SSSP tree, and the 
 hop-length $h_v (x)$ in $T_v$, i.e., 
 the number of edges on the path from $v$ to $x$.\label{algMain:h-hop-sssp-tree}
\State  Compute a \centers $Q$ of size $\Theta ((n \log n)/h)$ for the $h$-hop SSSP trees computed in Step~\ref{algMain:h-hop-sssp-tree} using Algorithm~\ref{algCB}.  \label{algMain:centers}
\State  {\bf for each} $c\in Q$  in sequence: 
compute SSSP with root $c$ and $\delta (c,v)$ at each $v\in V$ \label{algMain:bf}
\State {\bf for each} $c \in Q$ in sequence:  
broadcast to all other nodes 
$c$'s ID and
the $\delta_h(v,c)$ values it computed for each $v$  in Step~\ref{algMain:h-hop-sssp-tree}.	\label{algMain:broadcast}
\State {\bf Local step at each node:} At each $v$ compute $\delta(u,v)$ values for each $u \in V$ by using equation~\ref{eq:SD} with $\delta_h (u,v)$ from Step~\ref{algMain:h-hop-sssp-tree}, the $\delta(c,v)$ values from 
Step~\ref{algMain:bf},  and the $\delta_h (u,c)$ values received in Step~\ref{algMain:broadcast}. \label{algMain:local-compute}
\end{algorithmic} \label{algMain}
\end{algorithm}

  \begin{definition}[{\bf Blocker Set}~\cite{King99}]
  Let $\Tau$ be a collection of rooted $h$-hop trees in a graph $G=(V,E)$. A set $Q\subseteq V$ is a 
\emph{  \centers } for $\Tau$ if every root to leaf path 
 of length $h$ 
  in every tree in $\Tau$ contains a vertex in $Q$.
  Each vertex in $Q$ is called a \emph{blocker vertex} for $\Tau$.
  \end{definition}
  
  In Step 3 of Algorithm~\ref{algMain} we compute $\delta (c,v)$ for each $c\in Q$ and for all $v\in V$.
  In Step 4 each blocker vertex $c$ broadcasts all of the $\delta_h(v,c)$ values it computed
  in Step~\ref{algMain:h-hop-sssp-tree}.
  Finally, in Step~\ref{algMain:local-compute} each node
  $v$ computes $\delta(u,v)$ for each $u\in V$ using the values it computed or received in the
  earlier steps. More specifically, $v$ computes $\delta(u,v)$ as:
  \begin{equation}\label{eq:SD}
  \delta (u,v) = \min \left\{\delta_h(u,v), ~\min_{c\in Q} \left(\delta_h(u,c) + \delta (c,v)\right)\right\}
  \end{equation}
  
  \begin{lemma}
  The $\delta(u,v)$ values computed at each $v$ in Step~\ref{algMain:local-compute} of Algorithm~\ref{algMain} 
  are the correct shortest path distances.
  \end{lemma}
  
  \begin{proof}
  Fix vertices $u,v$ and consider a shortest path $p$ from $u$ to $v$. If $p$ has at most $h$ edges
  then $w(p) = \delta_h(u,v)$ and this value is directly computed at $v$ in Step~\ref{algMain:h-hop-sssp-tree}.
  Otherwise by the property of the
  \centers  $Q$ we know that there is a vertex $c\in Q$ which lies along $p$ within the
  $h$-hop SSSP tree rooted at $u$ that is constructed in Step~\ref{algMain:h-hop-sssp-tree}.
  Let $p_1$ be the portion of $p$ from $u$ to $c$ and let $p_2$ be the portion from $c$ to $v$.
  So $w(p_1) = \delta_h(u,c)$, $w(p_2)= \delta (c,v)$ and $w(p) = w(p_1) + w(p_2)$.

  The value $\delta_h(u,c)$ is received by $v$ in the broadcast step for center $c$ in Step~\ref{algMain:broadcast}.
  The value $\delta(c,v)$ is computed at $v$ when SSSP with root $c$ is computed in 
  Step~\ref{algMain:bf}. Hence $v$ has the information needed to compute $\delta(u,v)$ 
  in Step~\ref{algMain:local-compute} for each $u$ using Equation~\ref{eq:SD}.
  \end{proof}
  
  We now bound the number of rounds needed for each step in Algorithm~\ref{algMain} (other than
  Step~\ref{algMain:centers}). For this we first state bounds for some simple primitives that will be used
  to execute these steps.
  
  \begin{lemma}\label{lem:bf}
  Given a source $s\in V$, using the Bellman-Ford algorithm:\\
  (a) the shortest path distance $\delta (s,v)$ can be computed at each $v\in V$ in $n$ rounds.\\
  (b) the $h$-hop shortest path distance $\delta_h(s,v)$,
  the hop length $h_s(v)$, 
    and parent pointer 
  in the $h$-hop SSSP tree rooted at $s$ can be computed at each $v\in V$ in $h$ rounds.
  \end{lemma}
  
  \begin{lemma}\label{lem:k-bc}
  A node  $v$ can broadcast  $k$ local values to all other nodes 
    reachable from it
  deterministically in $O(n+k)$ rounds.
  \end{lemma}
  
  \begin{proof}
 We construct a BFS tree rooted at $v$ in at most $n$ rounds and then we pipeline the broadcast
  of the $k$ values. The root $v$ sends the $i$-th value to all its children
  in round $i$ for $1\leq i \leq k$. In a general
  round, each node $x$ that received a value in the previous round sends that value to all its children.
  It is readily 
  seen that the $i$-th value reaches all nodes at hop-length $d$ from $v$
 in the BFS tree  
   in round $i+d-1$, and 
  this is the only value that node $x$ receives in this round.
  \end{proof}
 
 \begin{lemma}\label{lem:all-to-all-bc}
  All $v\in V$ can broadcast a local value  to every other node 
they can reach
 in $O(n)$ rounds deterministically.
  \end{lemma}
  
   \begin{proof}
  This  broadcast can be done in $O(n)$ rounds in many ways, for example by
  piggy-backing on
  an $O(n)$ round unweighted APSP algorithm~\cite{LP13,PR18} (and 
  also \cite{HW12,PRT12} 
  for undirected graphs)
  where now each message contains the value sent
  by source $s$ in addition to the current shortest path distance estimate for source $s$.
  \end{proof}
 
 \begin{lemma}
Algorithm~\ref{algMain} runs in $O(n\cdot h + (n^2/h) \cdot \log n)$ rounds 
 assuming Step 2 can be implemented to run within this bound.
 \end{lemma}
 
 \begin{proof}
 Let the size of the blocker set be $q = \frac{n}{h}\cdot \log n$.
 Using part (b) of Lemma~\ref{lem:bf}, 
 Step 1 can be computed in $O(n \cdot h)$ rounds by computing the $n$ $h$-hop SSSP trees in sequence. Step 3 can be computed in $O(n \cdot q) = O((n^2/h) \cdot \log n)$ rounds
by part (a)  of Lemma~\ref{lem:bf}. Step 4 can be computed in $O(n \cdot q) = O((n^2/h) \cdot \log n)$ rounds
by Lemma~\ref{lem:k-bc} (using $k=n)$. Finally, Step 5 involves only local
computation and no communication. This establishes the lemma.
\end{proof}

 In the next section we give a deterministic algorithm to compute Step~\ref{algMain:centers} in
 $O(n\cdot h + (n^2/h) \cdot \log n)$ rounds, which leads to 
 our main theorem
 (by using
 $h = \sqrt{n \cdot \log n})$.
 
 \begin{theorem}
 Algorithm~\ref{algMain} is a deterministic distributed algorithm for weighted APSP
 in directed or undirected graphs 
  that runs in 
 $O(n^{3/2} \cdot \sqrt{\log n})$ rounds in the \congest{} model.
 \end{theorem}

\section{Computing a Blocker Set Deterministically}\label{sec:blocker}

  The simplest method to find a 
 \centers is to chose the vertices randomly.  
  An early use of
  this method for path problems in graphs was in
  Ullman and Yannakakis~\cite{UY91} where a random set of $O(\sqrt n \cdot \log n)$
   distinguished nodes
   was picked. It is readily seen that some vertex in this set will 
  intersect any path of $O(\sqrt n)$ vertices in the graph
 (and so this set would serve as a blocker set of size $O((n \log n)/h)$for our algorithm if $h= \sqrt n$). 
  Using this observation
   an improved randomized parallel algorithm
  (in the PRAM model) was given in~\cite{UY91} to compute the transitive closure. Since then this method of 
  using random sampling to choose
  a suitable \centers  has been used extensively in parallel and dynamic computation of transitive closure and shortest paths, and more recently, in distributed 
  computation of APSP~\cite{HNS17}.
  
  It is not clear if the above simple randomized strategy can be derandomized in its full generality. However, for our purposes
  a \centers only needs to intersect all paths in the set of hop trees we construct in Step~\ref{algMain:h-hop-sssp-tree} of Algorithm~\ref{algMain}. For this,
  a deterministic sequential algorithm for computing a \centers  was given in King~\cite{King99} in order to compute
  fully dynamic APSP. This algorithm computes a \centers  of size $O((n/h) \ln p)$
  for a collection $F$ of $h$-hop trees with a total of
  $p$ leaves across all trees (and hence $p$ root to leaf paths) in an $n$-node graph.
  In our setting $p\leq n^2$ since we have $n$ trees and each tree could have up to $n$ leaves.
  
  King's sequential blocker set algorithm uses the following simple observation: 
  Given a collection of $p$ paths each with exactly $h$ nodes from an underlying
  set $V$ of $n$ nodes, there must exist a vertex that is contained in at least
$ph/n$ paths. The algorithm adds one such vertex $v$ to the blocker set,
removes all paths that are covered by this vertex and repeats this process until no path remains in the collection. The number of paths is
reduced from $p$ to at most $(1-h/n) \cdot p$ when the blocker vertex $v$ is removed, hence after $O((n/h) \ln p)$ removals of
vertices, all paths are removed. Since $p$ is at most $n^2$ the size of the blocker set is $O((n \log n)/h)$. 
King's sequential algorithm for finding a blocker set
runs in $O(n^2 \log n)$ deterministic time.

We now describe our distributed algorithm to compute a blocker set.
As in King~\cite{King99},
for each vertex $v$ in a tree $T_x$  in the collection of trees $\Tau$ we define:

 \begin{itemize}
 \item  $score_x(v)$ is the number of leaves 
at depth $h$ in $T_x$ that are
in the subtree rooted at $v$ in $T_x$;
\item $score(v) = \sum_x score_x(v)$.
\end{itemize}
Thus, $score(v)$ is the number of 
root-to-leaf length paths of length $h$ in the collection of trees $\Tau$ that contain vertex $v$.
Initially, our distributed algorithm computes all $score_x(v)$ and $score(v)$ for all vertices $v\in V$ and all
$h$-hop trees $T_x$ in $O(n\cdot h)$ rounds using Algorithm~\ref{algCS}. Then through an all-to-all broadcast of $score(v)$ to all other nodes for all $v$,  all nodes identify the vertex $c$ with maximum score as the next 
blocker vertex to be removed from the
trees and added to the \centers  $Q$. (In case there are multiple vertices with the maximum score the algorithm chooses the
vertex of minimum id having this maximum score. This ensures that all vertices will locally choose the same vertex as the
next blocker vertex once they have received the scores of all vertices.) We repeat this process until all scores are zeroed out. 
By the discussion above (and as observed in~\cite{King99}) we will identify all 
 the vertices in $Q$
 in 
$O((n \cdot \log n)/h)$ 
repeats of this process. 

What remains is to obtain an $O(n)$ round procedure to update the
$score$  and $score_x$ values at all nodes each time a 
vertex $c$ is removed so that  we have the correct 
values at each node for each tree when the leaves covered by $c$ are removed from the tree.

If a vertex $v$ is a descendant of the removed 
vertex $c$ in $T_x$ then all paths in $T_x$ that pass
through $v$ are removed when $c$ is removed and hence $score_x(v)$ needs to go down to zero for each
such tree $T_x$ where $v$ is a descendant of 
the chosen blocker node
$c$. In order to faciliate an $O(n)$-round
computation of these updated $score_x$ values in each tree at all nodes that are descendants of 
$c$, we initially precompute at every node $v$ a list $Anc_x(v)$ all of its ancestors in each tree $T_x$.
This is computed in $O(n \cdot h)$ rounds using Algorithm~\ref{algAnc}. Thereafter, each time a
new blocker vertex $c$ is selected to be removed from the trees and added to $Q$, it is a local computation
at each node $v$ to determine which of the $Anc_x(v)$ sets at $v$ contain $c$ and to zero out $score_x(v)$ for each
such $x$.

The other type of vertices whose scores change after a vertex $c$ is removed are the ancestors of
$c$ in each tree. If $v$ is an ancestor of $c$ in $T_x$ then after $c$ is removed $score_x(v)$ needs
to be reduced by $score_x(c)$ (i.e., $c$'s score before it was removed and added to $Q$) since these paths no longer need
to be covered by $v$. For these ancestor updates we give an $O(n)$-round algorithm that runs after the
addition of each new blocker node to $Q$
 and correctly updates the scores for these ancestors in every tree
(Algorithm~\ref{algAU}). These algorithms together give the overall deterministic algorithm 
(Algorithm~\ref{algCB}) for the
computation of the \centers   $Q$  in $O(n \cdot h + (n^2 \log n)/h)$ rounds.

We now give the details of our algorithms.
  We assume that for a tree $T_s$ rooted at
$s$, each node $v$ in the tree knows $\delta(s,v)$, its shortest path distance from $s$, its hop length 
$h_{s}(v)$ (the
number of edges on the tree path from $s$ to $v)$, its parent node,
and all its children in $T_s$.

\subsection{The Blocker Set Algorithm}

Algorithm~\ref{algCB} gives our distributed deterministic method to compute a blocker set.
It uses a collection of helper algorithms that are described in the next section.
This blocker set algorithm
 is at the heart of our main algorithm (Algorithm~\ref{algMain}, Step~\ref{algMain:centers})
for computing the exact weighted APSP.

\begin{algorithm}[H]	
\caption{{\sc Compute-Blocker}}
Input: set $\Tau$ of all $h$-hop trees $T_x$; \hspace{0.1in} Output: set $Q$
\begin{algorithmic}[1]	
\State {\bf Initialization [lines~\ref{algCB:runAlgCS}-\ref{algCB:broadcast}]:} Run Algorithm~\ref{algCS} to compute scores for all $v \in V$	\label{algCB:runAlgCS}
\State For each $T_x$ compute the ancestors of each vertex $v$ in $T_x$ in $Anc_x(v)$ using Algorithm~\ref{algAnc}  \label{algCB:anc}
\For{each $v\in V$}\label{algCB:forloop}
\State {\bf Local Step:} $score(v) \leftarrow \sum_{x\in V} score_x (v)$	\label{algCB:local}
\State broadcast $\sv$ to all nodes in $V$ (using Lemma~\ref{lem:all-to-all-bc})	\vspace{0.05in} \label{algCB:broadcast}
\State {\bf Add blocker vertices to blocker set $Q$ [lines~\ref{algCB:startWhile}-\ref{algCB:endWhile}]:}
\State {\bf while} there is a node $c$ with $score (c) > 0$	\label{algCB:startWhile}
\State \hspace{.5in} {\bf Local Step:} select the node $c$ with max score as next vertex in $Q$
\label{algCB:select}
\State  \hspace{.5in} Run Algorithms~\ref{algDU} and \ref{algAU} to update $\sxv$ for each $x \in V$ and $\sv$	\label{algCB:updateScores}
\State \hspace{.5in} broadcast $\sv$ to all nodes in $V$ and receive $score(x)$ from all other nodes $x$	\label{algCB:endWhile}
\EndFor
\end{algorithmic} \label{algCB}
\end{algorithm}

Step~\ref{algCB:runAlgCS} of Algorithm~\ref{algCB} executes 
Algorithm~\ref{algCS}  to compute all the initial scores at all nodes $v$.
Step~\ref{algCB:anc} involves running Algorithm~\ref{algAnc} 
for pre-computing ancestors of each node in every $T_x$. 
Step~\ref{algCB:local} is a local computation  (no communication) where 
all nodes $v$ compute their total score by summing up the scores 
for all trees $T_x$ 
 to which they belong.
And in Step~\ref{algCB:broadcast},
each node $v$ broadcasts its score value to all other nodes.

The while loop in Steps~\ref{algCB:startWhile}-\ref{algCB:endWhile} of Algorithm~\ref{algCB} runs
as long as there is a  node with positive score.
In Step~\ref{algCB:select},
the node with maximum score is selected as the 
vertex $c$ to be added to $Q$
(and if there are multiple nodes with the maximum score, 
then among them the node with the minimum ID is selected,
so that the same node is selected locally at every vertex).
In Step~\ref{algCB:updateScores},
after 
blocker vertex $c$ is selected,
each
node $v$ checks whether it is a descendant of $c$
in each $T_x$
and if so update its score for that tree
using
 Algorithm~\ref{algDU}.
 This is followed by an execution of Algorithm~\ref{algAU} which updates the scores at each node $v$ for each
 tree $T_x$ in which $v$ is an ancestor of $c$.
Then in Step~\ref{algCB:endWhile},
all the nodes broadcast their score to 
 all other nodes so that they can all 
select 
the next
vertex to be added to $Q$.
This leads to the following lemma, assuming the results shown in the next section.

\bl \label{lemma:centers}
Algorithm~\ref{algCB} correctly computes the \centers  $Q$ in 
$O(n \cdot h + n \cdot |Q|)$ rounds.
\el

\begin{proof}
Step 1 runs in $O(n \cdot h)$ rounds (by Lemma~\ref{lemma:alg2Runtime}) and so does
Step~\ref{algCB:anc} (see Lemma~\ref{lemma:computeAnc} below). Step 4 is a 
local computation and the broadcast in Step~\ref{algCB:broadcast} runs
in $O(n)$ rounds by Lemma~\ref{lem:all-to-all-bc}. 

The while loop starting in Step~\ref{algCB:startWhile}
 runs for $|Q|$ iterations since a new
blocker
vertex is added to $Q$ in each iteration. In each iteration, Step 7 is a local computation as
is the execution of  Algorithm~\ref{algDU} in Step 8.
Algorithm~\ref{algAU}  in Step 8 runs in $O(n)$ rounds (Lemma~\ref{lemma:algAncestorUpdate}).
The all-to-all broadcast in Step 9 is the same as the initial all-to-all broadcast in
Step~\ref{algCB:broadcast}  and runs in $O(n)$ rounds. Hence each iteration of the
while loop runs in $O(n)$ rounds giving the desired bound.
\end{proof}

\subsection{Algorithms for Computing and Updating Scores}

In this section we give the details of our algorithms for computing initial scores (Algorithm~\ref{algCS})
and for updating these scores values once a blocker vertex $c$ 
is selected and added to the blocker set $Q$ (Algorithms~\ref{algAnc}-\ref{algAU}).

\begin{algorithm}[H]	
\caption{Compute Initial scores for a node $v$ in $T_x$}
\begin{algorithmic}[1]	
\State {\bf Initialization [Local Step]:} {\bf if} $h_x (v) = h$ {\bf then} $score_x (v) \leftarrow 1$ {\bf else} $score_x (v) \leftarrow 0$ 	\label{algCS:init}
\State {\bf In round $r > 0$:}
\State {\bf send:} {\bf if} $r = h - h_x (v) + 1$ {\bf then} send $\langle score_x (v) \rangle$ to $parent_x (v)$	\label{algCS:send} \vspace{0.05in}
\State {\bf receive [lines~\ref{algCS:newreceiveStart}-\ref{algCS:receiveEnd}]:}  
\If{$r = h - h_x(v)$}	\label{algCS:newreceiveStart}
	\State let $\mathcal{I}$ be the set of incoming messages to $v$ 
	\For{{\bf each} $M \in \mathcal{I}$} \label{algCS:receiveStart}
		\State let $M = \langle score^- \rangle$ and let the sender be $w$
	 	\State {\bf if} $w$ is a child of $v$ in $T_x$ {\bf then} $score_x (v) \leftarrow score_x (v) + score^-$ 
	\EndFor \label{algCS:receiveEnd}
\EndIf
\end{algorithmic} \label{algCS}
\end{algorithm}

\vspace{-0.1in}

Algorithm~\ref{algCS} gives the procedure for computing the initial scores for a node $v$ in a tree $T_x$.
In Step~\ref{algCS:init}
each leaf node 
at depth $h$
initializes its score for $T_x$ to $1$ and 
all other nodes
 set their initial score to $0$.
In a general round $r > 0$,
nodes with $h_x (v) = h+1-r$ send out their scores to their parents 
and nodes with $h_x(v) = h-r$ will receive all the scores from its children in $T_x$
and set its score equal to the sum of these received scores (Steps~\ref{algCS:newreceiveStart}-\ref{algCS:receiveEnd}).

\bl \label{lemma:alg2Runtime}
Algorithm~\ref{algCS} computes the initial scores for every node $v$ in $T_x$
in $O(h)$ rounds.
\el

\begin{proof}
The leaves at depth $h$ correctly initialize their score to $1$ locally in Step~\ref{algCS:init}.
Since we only consider paths of length $h$ from the root $x$ to a leaf, it
is readily seen that a node $v$ that is $h_x (v)$ hops away from $x$ in $T_x$ 
will receive scores from its children in round $h - h_x(v)$
and thus will have the correct $score_x (v)$ value to send in Step~\ref{algCS:send}.
\end{proof}

For every $x \in V$,
every node $v \in T_x$ will run this algorithm to compute their score in $T_x$.
Since every run of Algorithm~\ref{algCS} for a given $x$ takes $h$ rounds,
all the initial scores can be computed in $O(n\cdot h)$ rounds.

\vspace{-0.05in}
\begin{algorithm}[H]	
\caption{{\sc Ancestors $(v,x)$:} Algorithm for computing ancestors of node $v$ in $T_x$ at round $r$ }
\begin{algorithmic}[1]	
\State {\bf Initialization [Local Step]:} $Anc_x (v) \la \phi$ 
\State {\bf In round $r > 0$:}
\State {\bf send [lines~\ref{algAnc:sendStart}-\ref{algAnc:sendEnd}]:}
\If{$r = 1$}	\label{algAnc:sendStart}
	\State send $\langle v \rangle$ to $v$'s children in $T_x$	\label{algAnc:sendRound1}
\Else
	\State let $\langle y \rangle$ be the message $v$ received in round $r-1$	\label{algAnc:sendy1}
	\State send $\langle y \rangle$ to $v$'s children in $T_x$	\label{algAnc:sendy2}
\EndIf	\label{algAnc:sendEnd}
\State {\bf receive [lines~\ref{algAnc:receiveStart}-\ref{algAnc:receiveEnd}]:}
\State let $\langle y \rangle$ be the message $v$ received in this round	\label{algAnc:receiveStart}
\State add $y$ to $Anc_x (v)$	\label{algAnc:receiveEnd}
\end{algorithmic} \label{algAnc}
\end{algorithm}

Algorithm~\ref{algAnc} describes our algorithm for 
precomputing
the ancestors of each node $v$ in a
tree $T_x$
of height $h$.
In round 1, 
every node $v$ sends its ID to its children in $T_x$ as described in Step~\ref{algAnc:sendRound1}.
And in a general round $r$,
$v$ sends the ID of the ancestor that it received in round $r-1$ (Steps~\ref{algAnc:sendy1}-\ref{algAnc:sendy2}).
If a node $v$ receives the ID of an ancestor $y$,
then it immediately adds it to its ancestor set, $Anc_x (v)$ (Steps~\ref{algAnc:receiveStart}-\ref{algAnc:receiveEnd}).

\bl \label{lemma:computeAnc}
For a tree $T_x$
of height $h$
 rooted at vertex $x$, 
Algorithm~\ref{algAnc} correctly computes the set of ancestors for all nodes $v$ in $T_x$ in $O(h)$ rounds. 
\el

\begin{proof}
We show that all nodes $v$ correctly computes all their ancestors in $T_x$ in the set $Anc_x (v)$ 
using induction on round $r$.
We show that by round $r$,
every node $v$ has added all its ancestors that are at most $r$ hops away from $v$.

If $r = 1$, then $v$'s parent in $T_x$ (say $y$) would have send out its $ID$ to $v$ in Step~\ref{algAnc:sendRound1}
and $v$ would have added it to $Anc_x (v)$ in Step~\ref{algAnc:receiveEnd}.

Assume that every node $v$ has already added all ancestors in the set $Anc_x (v)$ 
that are at most $r-1$ hops away from $v$.

Let $u$ be the ancestor of $v$ in $T_x$ that is exactly $r$ hops away from $v$.
Then by induction, 
$u \in Anc_x (y)$ since $u$ is exactly $r-1$ hops away from $y$
and thus $y$ must have send $u$'s $ID$ to $v$ in round $r$ in Step~\ref{algAnc:sendy2} 
and hence $v$ would have added $u$ to its set $Anc_x (v)$ in
round $r$ in Step~\ref{algAnc:receiveEnd}. 
\end{proof}

\begin{algorithm}[H]	
\caption{Algorithm for updating scores at $v$ when $v$ is a descendant of new blocker node $c$}
Input: blocker vertex $c$ added to $Q$.\\
There is no communication in this algorithm, it is entirely a {\bf local computation} at 
$v$.
\begin{algorithmic}[1]	
\If{$score (v) \neq 0$}
	\For{{\bf each} $x \in V$}
		\If{$c \in Anc_x (v)$}	\label{algDU:check}
			\State $score (v) \la score (v) - score_x (v)$	\label{algDU:update1}
			\State $score_x (v) \la 0$	\label{algDU:update2}
		\EndIf
	\EndFor
\EndIf
\end{algorithmic} \label{algDU}
\end{algorithm}

Once we have pre-computed the $Anc_x(v)$ sets for all vertices $v$ and all trees $T_x$ using
Algorithm~\ref{algAnc}, updating the scores at each node  for all trees in which it is a descendant
of the newly chosen blocker node $c$ becomes a purely local computation.
Algorithm~\ref{algDU} describes the algorithm at node $v$ that updates its scores
after  a vertex $c$ is 
added as a blocker node to $Q$.
At node $v$ for each given $T_x$,
$v$ checks if $c \in Anc_x (v)$
and if so update its score values in Steps~\ref{algDU:update1}-\ref{algDU:update2}.

\bl \label{lemma:algDescendantUpdate}
Given a 
blocker vertex $c$, 
Algorithm~\ref{algDU} correctly updates the scores of all nodes $v$ 
such that $v$ is a descendant of $c$ in some tree $T_x$.
\el

\begin{proof}
Fix a vertex $v$ 
and a tree $T_x$ 
such that $v$ is a descendant of $c$ in $T_x$.
By Lemma~\ref{lemma:computeAnc} $c \in Anc_x (v)$,
and thus $v$ will correctly update its score values in Steps~\ref{algDU:update1}-\ref{algDU:update2}.
\end{proof}

We now move to the last remaining part of the blocker set algorithm: our method to
correctly update scores at ancestors of the newly chosen blocker node $c$ in each $T_x$. Recall
that if $v$ is an ancestor of $c$ in $T_x$ we need to subtract $score_x(c)$ from $score_x(v)$.
Here, in contrast to Algorithms~\ref{algAnc} and \ref{algDU} for nodes that are descendants of $c$ in
a tree, we do not precompute anything. Instead we give an $O(n)$-round method in Algorithm~\ref{algAU}
to correctly update scores for each vertex for all trees in which that vertex is an ancestor of $c$.

Before we describe Algorithm~\ref{algAU} we establish the following lemma, which is key to our
$O(n)$-round method.

\bl \label{lemma:intree}
Fix a vertex $c$.
For each root vertex $x \in V -\{c\}$,
let $\pi_{x,c}$ be the path from $x$ to $c$ in the $h$-hop SSSP tree $T_x$. 
Let $T = \cup_{x \in V-\{c\}} \{e ~|~e$ lies on  $\pi_{x,c}\}$,
i.e., $T$ is the set of edges that lie on some $\pi_{x,c}$.
Then $T$ is an in-tree rooted at $c$.
\el

\begin{proof}
If not,  there exists some $x, y \in V-\{c\}$ 
such that $\pi_{x,c}$ and $\pi_{y,c}$ coincide first at some vertex $z$ 
and  the subpaths in $\pi_{x,c}$ and $\pi_{y,c}$ from $z$ to $c$ are different.

Let these paths coincide again at some vertex $z'$ (such a vertex exists since their endpoint is same)
after diverging from $z$.
Let the subpath from $z$ to $z'$ in $\pi_{x,c}$ be $\pi^1_{z,z'}$
and the corresponding subpath in $\pi_{y,c}$ be $\pi^2_{z,z'}$.
Similarly let $\pi_{x,z}$ be the subpath of $\pi_{x,c}$ from $x$ to $z$
and let $\pi_{y,z}$ be the subpath of $\pi_{y,c}$ from $y$ to $z$.

Clearly both $\pi^1_{z,z'}$ and $\pi^2_{z,z'}$ have equal weight
(otherwise one of $\pi_{x,c}$ or $\pi_{y,c}$ cannot be a shortest path).
Thus the path $\pi_{x,z}\circ \pi^2_{z,z'}$ is also a shortest path.

Let $(a,z')$ be the last edge on the path $\pi^{1}_{z,z'}$ 
and $(b,z')$ be the last edge on the path $\pi^{2}_{z,z'}$.

Now since the path $\pi_{x,z'}$ has $(a,z')$ as the last edge 
and we break ties using the IDs of the vertices,
hence $ID(a) < ID(b)$.
But then the shortest path $\pi_{y,z'}$ must also have chosen $(a, z')$ as the last edge 
 and hence $\pi_{y,z}\circ \pi^{1}_{z,z'}$ must be the subpath of path $\pi_{y,c}$,
 \rco
 \end{proof}

\begin{algorithm}[H]	
\caption{Pipelined Algorithm for updating scores at $v$ for all trees $T_x$ in which $v$ is an ancestor of newly chosen blocker node $c$}
Input: current \centers $Q$, newly chosen blocker node $c$
\begin{algorithmic}[1]	
\State {\bf Send [lines~\ref{algAU:initStart}-\ref{algAU:initEnd}]: (only for  $c$)}
\State {\bf Local Step at $c$:} create a list $list_c$ and {\bf for each}  $x \in V$ {\bf do} add an entry $Z = \langle x, score_x (c) \rangle$ to $list_c$ if $score_x (c) \neq 0$; then set $score_x (c)$ to $0$ for each $x \in V$ and set $score (c)$ to $0$	\label{algAU:initStart}
\State {\bf Round $i$:} let $Z = \langle x, score_x (c) \rangle$ be the $i$-th entry in $list_c$; send $\langle Z \rangle$ to $c$'s parent in $T_x$ \label{algAU:initEnd}
\State {\bf In round $r > 0$: (for vertices $v \in V - Q -\{c\}$)} \vspace{0.05in}
\State {\bf send [lines~\ref{algAU:sendStart}-\ref{algAU:sendEnd}]:}
\If{$v$ received a message in round $r-1$}	\label{algAU:sendStart}
	\State let that message be $\langle Z \rangle = \langle x, score_x (c) \rangle$.	
	\State { \bf if} $v \neq x$ {\bf then} send $\langle Z\rangle$ to $v$'s parent in $T_x$ \label{algAU:sendEnd} \vspace{0.05in}	
\EndIf
\State {\bf receive [lines~\ref{algAU:receiveStart}-\ref{algAU:receiveEnd}]:}
\If{$v$ receives a message $M$ of the form $\langle x, score_x (c) \rangle$}	\label{algAU:receiveStart}
	\State $\sxv \leftarrow \sxv - score_x (c)$; $\sv \leftarrow \sv - score_x (c)$	\label{algAU:update}
\EndIf	\label{algAU:receiveEnd}
\end{algorithmic} \label{algAU}
\end{algorithm}

Lemma~\ref{lemma:intree} allows us to re-cast the task for ancestor nodes to the following (where
we use the notation in the statement of Lemma~\ref{lemma:intree}):
the new blocker node $c$ needs to send $score_x(c)$ to all nodes on $\pi_{x,c}$ for each tree
$T_x$. Recall that in
the \congest{} model for directed graphs the graph edges are bi-directional. Hence this task can
be  accomplished by having $c$ send out $score_x(c)$ for each tree $T_x$ (other than $T_c$)
in $n-1$ rounds, one score
per round (in no particular order) along the parent edge for $T_x$. Each message 
$\langle x, score_x(c)\rangle$ will move along edges in $\pi_{x,c}$ (in reverse order) along parent edges in $T_x$ from $c$ to $x$.
Consider any node $v$. In general it will be an ancestor of $c$ in some subset of the $n-1$ trees
$T_x$. But the characterization in Lemma~\ref{lemma:intree} establishes that the incoming edge to
$v$ in all of these trees is the same edge $(u,v)$ and this is the unique edge on the path from $c$ to
$v$ in tree $T$ ($T$ is defined in the statement of Lemma~\ref{lemma:intree}). In fact, the  
messages for all of the trees in which $v$ is an ancestor of $c$ will traverse exactly the same path from
$c$ to $x$. Hence,  for the messages sent out by $c$ for the different trees in $n-1$ different rounds
(one for each tree other than $T_c$), if each vertex simply forwards any message
$\langle x, score_x(c)\rangle$  it receives to its parent
in tree $T_x$ all messages will be pipelined to all ancestors in $n-1 +h$ rounds. This  is what is
done in Algorithm~\ref{algAU}, whose steps we describe below, for completeness.

Step~\ref{algAU:initStart} of Algorithm~\ref{algAU} is local computation at the new  blocker vertex $c$
where for each $T_x$ to which $c$ belongs,
$c$ adds an  entry  $\langle x, score_x (c) \rangle$ to a local list $list_c$.
In round $i$, 
$c$ sends the $i$-th entry in its list, say $\langle y, score_y(c) \rangle$, to its parent in $T_y$.
For node $v$ other than $c$,
in a general round $r > 0$,
if $v$ receives a message for some $x \in V$
it updates its score value for $x$ (Steps~\ref{algAU:receiveStart}-\ref{algAU:update})
and then forwards this message to its parent in $T_x$ in round $r+1$ (Step~\ref{algAU:sendStart}-\ref{algAU:sendEnd}).

 \bl \label{lemma:algAncestorUpdate}
Given a new
blocker vertex $c$, 
Algorithm~\ref{algAU} correctly updates the scores of all nodes $v$ 
in every tree $T_x$ in which $v$ is an ancestor of $c$ in $O(n+h)$ rounds.
\el

\begin{proof}
Correctness of Algorithm~\ref{algAU} was argued above. For the number of rounds, 
$c$ sends out it last message in round $n-1$, and if $\pi_{v,c}$ has length $k$ then
$v$ receives all messages sent to it by round $n-1+k$. Since we only have $h$-hop trees
$k \leq h$ for all nodes, and the lemma follows.
\end{proof}

\section{Conclusion}\label{sec:concl}

We have presented a new distributed algorithm for 
the  exact computation of weighted all pairs shortest paths in
 both directed and undirected graphs.
This algorithm runs in $O(n^{3/2} \cdot \sqrt{\log n})$ rounds and is the first $o(n^2)$-round
deterministic algorithm for 
this problem in the {\sc Congest} model.
At the heart of our algorithm is a deterministic algorithm for
computing blocker set.
Our blocker set construction may have applications in other distributed algorithms that need to identify
a relatively small set of vertices that intersect all paths in a set of paths with the same (relatively long) length.

The main open question left by our work is to improve the round-bound for deterministic weighted APSP.
In the unweighted case APSP can be computed in $O(n)$ rounds~\cite{HW12,PRT12,LP13,PR18}, and weighted APSP can be computed in
$\tilde{O}(n^{5/4})$ rounds~\cite{HNS17} w.h.p. with randomization. Considering the simplicity of our
algorithm and the dramatic improvement over the previous (trivial) bound, we believe that further 
improvements could be achievable.
 Also of independent interest is to explore better distributed algorithms to find
a blocker set.

Very recently in~\cite{AR18},
a new pipelined algorithm for Step~\ref{algMain:h-hop-sssp-tree} of our Algorithm~\ref{algMain}
was presented that runs in $O(n\sqrt{h})$ rounds (improved from $O(nh)$ that we have here).
This, together with some changes to the blocker set algorithm in Step~\ref{algMain:centers},
gives an $\tilde{O}(n^{4/3})$ rounds deterministic algorithm for APSP in the Congest model.

\newpage
\bibliographystyle{abbrv}
\bibliography{references}

\end{document}